\documentclass[letterpaper, 10 pt, conference]{ieeeconf}

\IEEEoverridecommandlockouts  \usepackage[font=footnotesize]{caption}

\usepackage{graphicx}
\usepackage{amsmath,amssymb,amsthm}
\usepackage{mathrsfs}
\usepackage{cite}
\usepackage{url}
\usepackage{float}

\newtheorem{theorem}{Theorem}
\newtheorem{proposition}{Proposition}

\newcommand{\mca}{\mathcal{A}}
\newcommand{\mcd}{\mathcal{D}}

\newcommand{\ea}{\mathrm{ea}}
\newcommand{\ei}{\mathrm{ei}}
\newcommand{\ep}{\mathrm{ep}}

\title{Information Structure of Defection Decisions in a Colonel Blotto Model of Deterrence}

\author{Tristan Mott, David Grimsman, Keith Paarporn
\thanks{An abridged version of this work has been accepted for presentation at the 65th IEEE Conference on Decision and Control (CDC 2026). This work has been submitted to the IEEE for possible publication. Copyright may be transferred without notice, after which this version may no longer be accessible. T. Mott and D. Grimsman are with the Department of Computer Science at Brigham Young University. K. Paarporn is with the Department of Computer Science at the University of Colorado, Colorado Springs. This work is partially supported by NSF grant \#2346791 and \#2541011. Contact: \texttt{tristanmott15@gmail.com, grimsman@cs.byu.edu, kpaarpor@uccs.edu}}}

\begin{document}

\maketitle

\begin{abstract}
In strategic interactions, deterrence is an outcome where one party incentivizes the other party to not participate. This concept can be applied to Colonel Blotto games between two generals where troops can be deterred from following orders by the actions of the opposing general. In this setting, deterrence depends not only on how generals allocate troops, but also on the information available to troops when they decide whether to follow orders or defect. In this work, we study a two-battlefield Colonel Blotto model of deterrence under multiple information structures, ranging from decisions made with only aggregate information to decisions made after local battlefield conditions are observed. 

We show how different information structures alter the structure of the induced game, affect the existence of equilibria, and substantially change expected utilities. Our findings also prove that the ratio of troops available to each general combined with a capture probability threshold is enough to inform generals about what information they want their troops to have access to.
\end{abstract}

\section{Introduction}

Colonel Blotto games model strategic interactions in which players allocate limited resources across multiple contests with the goal of winning as many contests as possible \cite{gross1950continuous,roberson2006colonel}. Since their introduction more than a century ago \cite{borel1921theorie}, they have generated a broad literature \cite{gross1950continuous,roberson2006colonel}. Part of their appeal is that their simple allocation rules give rise to complex strategy spaces and intractable equilibria. Many variants remain unsolved, and researchers often approach them through continuous approximations, General Lotto relaxations, structural generalizations, and computational methods \cite{hart2008discrete,kovenock2021generalizations,ahmadinejad2019duels}.

These games are often used to study allocation problems in settings such as military planning, security, and related strategic competitions \cite{pita2009using}. In these settings, the allocated resources may represent military units or other human actors, and the contests may represent targets or battlefields \cite{shubik1981systems}. Human actors do not always follow orders, introducing a layer of strategic complexity that is not captured by the standard Blotto model. In such settings, a realistic strategic model should also account for deterrence \cite{schelling1960strategy,powell1990nuclear,original_blotto_deterrence}.

We define deterrence as the strategic behavior of \emph{attacking} troops who may choose to defect instead of following orders, especially when following orders is associated with some perceived risk, such as the possibility of being captured by the defending troops. Deterrence in security settings modeled by Colonel Blotto games is important for multiple reasons. First, understanding the deterrence dynamics can help the defender design strategies that exploit the possibility that attacking troops may defect, thereby increasing expected utility. Second, inducing attacking troops to defect can reduce the number of direct interactions between attackers and defenders, which in many settings can save lives. Thus, understanding how to design strategies that induce deterrence can be of critical importance in real-world applications.

Previous work \cite{original_blotto_deterrence} introduced a model of deterrence in a two-battlefield Colonel Blotto game, where attacking troops have the option to defect based on their perceived probability of being captured by defending troops. The defending and attacking generals then choose their high-level strategies in a way that anticipates the defection decisions of the attacking troops. That paper assumed attacking troops observe the realized pure allocations before deciding whether to defect, but do not get to observe their individual battlefield assignments.

This paper extends that work by considering how the information available to attacking troops affects the deterrence dynamics and subsequent equilibria. Other research has addressed incomplete, asymmetric, strategically revealed, or dynamically evolving information in the context of Blotto and Lotto games \cite{kovenock2011multidim,paarporn2024intentions,paarporn2025incomplete,anbarci2023dynamicblotto}. These works change what the high-level players, or generals, know about resources, intentions, valuations, or prior moves. By contrast, rather than studying the information available to the generals, we study the information available to the attacking troops. We then examine how that information affects deterrence and the induced game between the generals. We assume the attacking troops have a common capture threshold, and that they will defect if their perceived probability of capture exceeds this threshold. The perceived capture probability depends on the information available, which we model across three different levels: \emph{ex-ante}, \emph{ex-interim} (same as in the original paper \cite{original_blotto_deterrence}), and \emph{ex-post}.

Our main contributions are as follows.
\begin{itemize}
    \item We formulate the ex-ante deterrence model and identify the threshold quantities $\Gamma$ and $\Lambda$, between which the game can fail to admit a mixed-strategy equilibrium.
    \item We formulate the ex-post deterrence model as a finite zero-sum game, characterize the structure of its payoff matrix, and compare it directly with the ex-interim model.
    \item We show that comparing the attacking troops' capture threshold, $\theta$, with the ratio of defending to attacking troops, $N_d/N_a$, is enough for a general to determine how much information they want their troops to have access to.
\end{itemize}

The remainder of the paper is organized as follows. Section II presents the model and the different information structures. Section III studies the ex-ante model. Section IV studies the ex-post model. Section V presents numerical comparisons between these new models and the original ex-interim model. Section VI concludes.

\section{Model}

We define the baseline two-battlefield Blotto game, introduce troop deterrence, and derive the induced games under the ex-ante, ex-interim, and ex-post information structures.

\subsection{Baseline Colonel Blotto Game}

An attacking general $\mca$ has $N_a > 0$ troops, and a defending general $\mcd$ has $N_d > 0$ troops.
They decide how to allocate their troops over two battlefields, which we refer to as ``left'' and ``right''.
A pure strategy for the attacker is any integer $i \in \{0,1,\dots,N_a\}$,
meaning that $i$ troops are sent to the left battlefield and $N_a-i$ are sent to the right battlefield.
In a similar fashion, a pure strategy for the defender is any integer $j \in \{0,1,\dots,N_d\}$.

For the pair of strategies $(i,j)$, the attacker wins the left battlefield if and only if $i>j$, and wins the right battlefield if and only if $N_a-i > N_d-j$.
The payoff to the defender is given by

\begin{equation}
    B_{i,j} := 1-\mathbf{1}\{i>j\}-\mathbf{1}\{N_a-i > N_d-j\}
\end{equation}
where $\mathbf{1}\{\cdot\}$ denotes the indicator function.
In words, $B_{i,j}=1$ if the attacker wins neither battlefield, $B_{i,j} = 0$ if each wins one, and $B_{i,j} = -1$ if the attacker wins both.
We denote $B$ to be the $(N_a+1)\times (N_d+1)$ matrix of entries $B_{i,j}$.

The generals are able to randomize over their strategies.
A mixed strategy for the attacker and defender is any $X \in \Delta^{N_a+1}$ and $Y\in \Delta^{N_d+1}$, respectively,
where $\Delta^k$ denotes the $k$-dimensional probability simplex. 

The utility function to the defender of the baseline game is
\begin{equation}
    U_\mcd(X,Y) := X^\top B Y,
    \label{eq:baseline_expected_utility}
\end{equation}
and the attacker's utility function is then given by $U_\mca(X,Y) := -U_\mcd(X,Y)$.
This formulation is known as a \emph{Colonel Blotto game}, which we refer to with $\text{CB}(N_a,N_d)$.

\subsection{Troop Deterrence Rule}

We extend the baseline game by allowing each attacking troop to decide whether to follow its assigned order or \emph{defect}.
In the baseline game, all troops follow their orders unconditionally.

We model each battlefield as being compromised, and therefore won by the attacker, if at least one attacking troop makes it through uncaptured.
The defender's role is therefore to capture attacking troops, while the attacker's role is to get troops through.
In this formulation, a defending troop is able to capture exactly one attacking troop.
Consequently, the probability of capture that an attacking troop faces, given that it was assigned to a battlefield contested by $a$ attacking troops (including itself) and $d$ defending troops,
is given by 
\begin{equation}\label{eq:pc}
    p_\text{c}(a,d) = \min\left\{\frac{d}{a},1\right\}.
\end{equation}
This quantity only applies to battlefields with assigned attacking troops, so $a\ge 1$.
The conditional decision to defect is based on an attacking troop's \emph{perceived} probability of getting captured by the defending troops.
The information available to a troop determines its perceived capture probability. We consider three information structures:
\begin{itemize}
    \item \emph{ex ante}: troops observe only the mixed strategies $X,Y$ of the two generals;
    \item \emph{ex interim}: troops observe the realized pure strategies $i,j$, but not their individual battlefield assignments;
    \item \emph{ex post}: troops observe the realized pure strategies as well as their individual battlefield assignments.
\end{itemize}

We will assume all attacking troops have a \emph{capture threshold} of $\theta \in [0,1)$.\footnote{In the previous paper \cite{original_blotto_deterrence}, this threshold was written as $\theta = m/p$, where $m$ is the reward from attacking and $p$ is the penalty from capture.} 
Throughout the paper, we assume a troop defects if its perceived capture probability is greater than $\theta$, and otherwise attacks; in particular, troops who are indifferent can be convinced by their general to attack, so deterrence occurs only when the perceived capture probability is strictly greater than $\theta$. 
We also assume that all attackers exposed to the same perceived capture probability make the same decision; as shown in the original paper \cite{original_blotto_deterrence}, this is always a Nash equilibrium, and it is the natural symmetric outcome in that setting.
The generals choose their mixed strategies knowing the information structure and anticipating the resulting defection decisions.

\subsection{Ex-Ante Decision Point}

For the ex-ante information structure, attacking troops observe only the generals' mixed strategies $X,Y$, so their perceived capture probability is the expected capture fraction over all pure allocations. Let $A \in [0,1]^{(N_a+1)\times(N_d+1)}$ be the capture fraction matrix, where $A_{i,j}$ is the fraction of attacking troops that would be captured under pure strategy pair $(i,j)$ if all attacking troops followed their orders.
Equivalently, $A_{i,j}$ is the expected capture probability of an attacking troop sampled uniformly at random:
\begin{equation}\label{eq:Aij}
A_{i,j}
:= \frac{1}{N_a}\left(\min\{i,j\}+\min\{N_a-i,N_d-j\}\right).
\end{equation}

If $I\sim X$ and $J\sim Y$ are the independently sampled pure strategies, then the common perceived capture probability is
\begin{equation}
    p_{\ea}(X,Y) := \mathbb{E}\!\left[A_{I,J}\right] = X^\top A Y.
\end{equation}

Thus, the defender's utility is
\begin{equation}\label{eq:ex_ante_utility}
    U_{\mcd,\ea}(X,Y)
    :=
    \begin{cases}
        X^\top B Y, & X^\top A Y\le \theta,\\[4pt]
        1, & X^\top A Y>\theta,
    \end{cases}
\end{equation}
with attacker utility $-U_{\mcd,\ea}(X,Y)$. Denote this game by $\text{CB}_\ea(N_a,N_d,\theta)$.

\subsection{Ex-Interim Decision Point}

In the ex-interim information structure \cite{original_blotto_deterrence}, the attacking troops observe the sampled pure strategies $(i,j)$, but not their battlefield assignments. Each troop therefore averages over the two possible battlefield assignments induced by $(i,j)$, giving perceived capture probability
\begin{equation}
    p_{\ei}(i,j) := A_{i,j}.
\end{equation}
Equivalently, $A_{i,j}$ is the perceived capture probability of an attacking troop that observes $(i,j)$ but not its battlefield assignment.
We can then formulate a payoff matrix $M_\ei \in \mathbb{R}^{(N_a+1)\times(N_d+1)}$ with entries \cite{original_blotto_deterrence}
\begin{equation}\label{eq:ex_interim_matrix}
\begin{aligned}
M_{\ei,i,j}
&=
1
-
\mathbf{1}\!\left\{A_{i,j}\le \theta \ \text{and}\ i>j\right\} \\
&\quad
-
\mathbf{1}\!\left\{A_{i,j}\le \theta \ \text{and}\ N_a-i>N_d-j\right\}.
\end{aligned}
\end{equation}
Hence, the defender's utility is
\begin{equation}\label{eq:ex_interim_utility}
    U_{\mcd,\ei}(X,Y) := X^\top M_\ei Y,
\end{equation}
with attacker utility $-U_{\mcd,\ei}(X,Y)$. Denote this game by $\text{CB}_\ei(N_a,N_d,\theta)$.

\subsection{Ex-Post Decision Point}

In the ex-post information structure, each attacking troop observes the sampled pure strategies $(i,j)$ as well as the battlefield to which it is assigned.
Unlike the previous two information structures, the attacking troops can form different perceived capture probabilities and thus may make different decisions.
In particular, those sent to the left battlefield will perceive it as $p_\text{c}(i,j)$, while those sent to the right battlefield will perceive it as $p_\text{c}(N_a-i,N_d-j)$.
Because $\theta<1$, attackers at a battlefield choose to attack only when at least one assigned attacker would remain uncaptured. Thus, the attacking troops win exactly those battlefields where an assigned troop attacks and remains uncaptured.

From this, we can then formulate a payoff matrix $M_\ep \in \mathbb{R}^{(N_a+1)\times(N_d+1)}$ with entries
\begin{equation}\label{eq:ex_post_matrix}
\begin{aligned}
    M_{\ep,i,j}
    &=
    1
    -
    \mathbf{1}\!\left\{i>0 \ \text{and}\ j \le \theta i \right\} \\
    &\quad
    -
    \mathbf{1}\!\left\{N_a-i>0 \ \text{and}\ j \ge N_d-\theta(N_a-i) \right\}.
\end{aligned}
\end{equation}
Hence, the defender's utility is
\begin{equation}\label{eq:ex_post_utility}
    U_{\mcd,\ep}(X,Y) := X^\top M_\ep Y,
\end{equation}
with attacker utility $-U_{\mcd,\ep}(X,Y)$. Denote this game by $\text{CB}_\ep(N_a,N_d,\theta)$.

\subsection{Equilibrium Concept}

In all formulations $\text{CB}_\ea(N_a,N_d,\theta)$, $\text{CB}_\ei(N_a,N_d,\theta)$, and $\text{CB}_\ep(N_a,N_d,\theta)$, we are interested in characterizing the equilibrium payoff to the defender.
For $\ell \in \{\ea,\ei,\ep\}$, an equilibrium is a pair of mixed strategies $(X^*,Y^*)\in \Delta^{N_a+1}\times \Delta^{N_d+1}$ for which 
\begin{equation}
    \begin{aligned}
        U_{\mcd,\ell}(X^*,Y) \leq U_{\mcd,\ell}(X^*,Y^*) \leq U_{\mcd,\ell}(X,Y^*)
    \end{aligned}
\end{equation}
for all $X \in \Delta^{N_a+1}$ and $Y\in \Delta^{N_d+1}$.

\subsection{Simple Example}

Consider an example where both generals have two troops. 
In the baseline game $\text{CB}(2, 2)$, the equilibrium strategies are given by $X^*=Y^*=[1/3,1/3,1/3]^\top$.
Suppose they utilize these strategies, and suppose the pure strategy samples are $(i, j) = (1, 2)$.

\begin{figure}[H]
    \centering
    \includegraphics[width=\columnwidth]{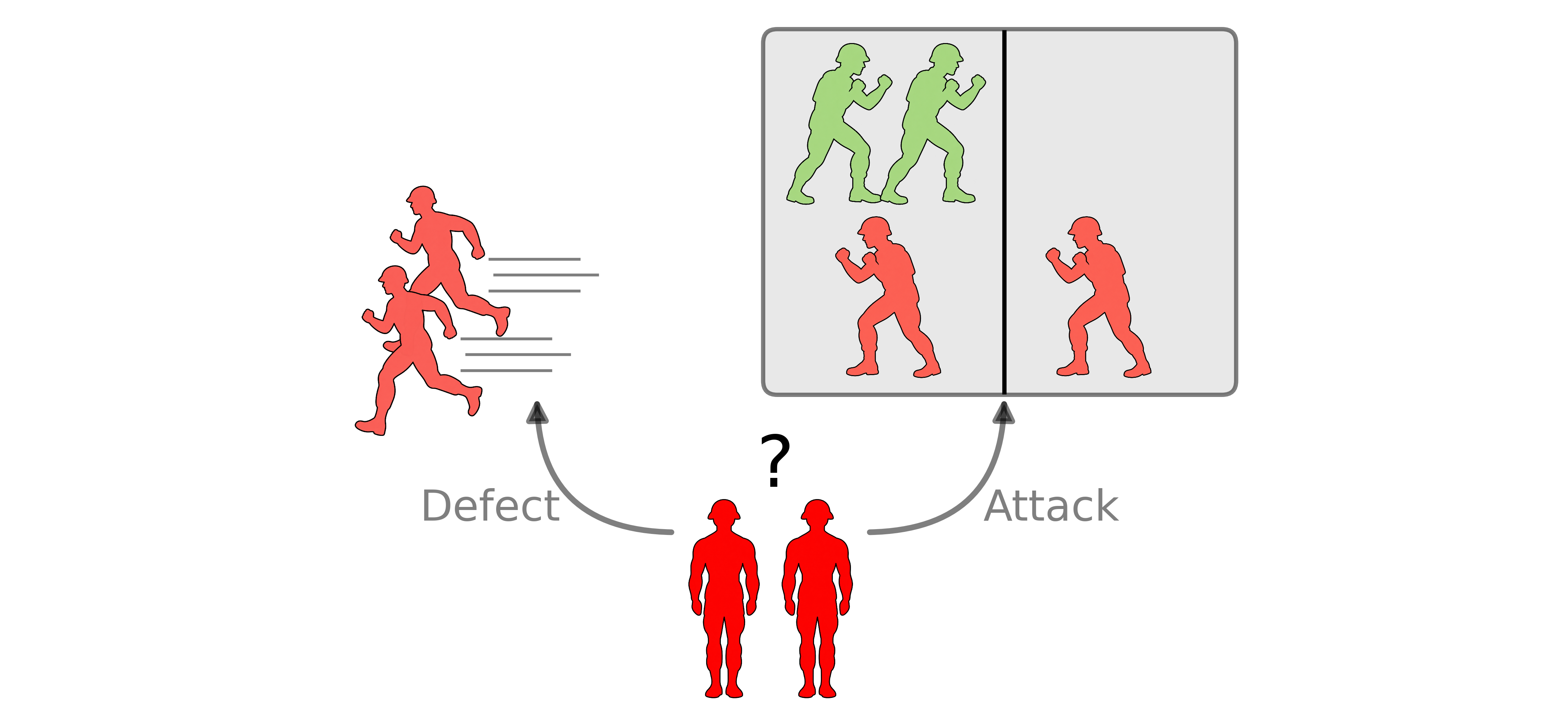}
    \caption{Example allocations for $N_a=2$, $N_d=2$.}
    \label{fig:allocation_example}
\end{figure}

In this example, the ex-ante perceived capture probability is $p_\ea(X,Y)=5/9$,
the ex-interim perceived capture probability is $p_\ei(1,2)=A_{1,2}=1/2$,
and the ex-post perceived capture probability is $p_\text{c}(1,2)=1$ for the troop on the left,
and $p_\text{c}(1,0)=0$ for the troop on the right.
Depending on the value of $\theta$, the information level may therefore have a significant impact on the outcome of the game.

\section{Ex-Ante Decision Point}

We now study the ex-ante model, in which attacking troops decide whether or not to defect after observing only the generals' mixed strategies. The defender's payoff is defined by the piecewise function \eqref{eq:ex_ante_utility}; if the perceived capture probability is above the threshold, all attacking troops defect, and otherwise the baseline Colonel Blotto game is played. Because the payoff depends jointly on the baseline utility matrix, the capture fraction matrix, and the deterrence threshold, it is no longer bilinear in the mixed strategies.

Let $\mathcal{E}_B$ denote the set of attacker equilibrium strategies in the baseline game with payoff matrix $B$.

\begin{theorem}
Define
\begin{equation}
\Gamma := \max_Y \min_X X^\top A Y
\label{eq:Gamma_def}
\end{equation}
and
\begin{equation}
\Lambda := \min_{X_B\in\mathcal{E}_B}\max_Y X_B^\top A Y.
\label{eq:Lambda_def}
\end{equation}
Then:
\begin{enumerate}
\item If $\theta < \Gamma$, then there exists a defender strategy $Y^*$ such that $(X,Y^*)$ is a Nash equilibrium of the ex-ante game for every attacker strategy $X$, and deterrence occurs in equilibrium.
\item If $\theta \ge \Lambda$, then there exists a baseline equilibrium attacker strategy $X_B\in\mathcal{E}_B$ such that $(X_B,Y_B)$ is a Nash equilibrium of the ex-ante game for any defender best response $Y_B$ to $X_B$ in the baseline game.
\item If $\Gamma \le \theta < \Lambda$, then there exist game instances for which no mixed strategy Nash equilibrium exists.
\end{enumerate}
\end{theorem}

In words, $\Gamma$ is the highest expected capture fraction the defender can guarantee against any attacker strategy, while $\Lambda$ is the highest expected capture fraction the defender can guarantee when the attacker is restricted to baseline equilibrium strategies. The defender can therefore force deterrence when $\theta<\Gamma$ and the attacker can safely use a baseline equilibrium when $\theta\ge\Lambda$. The intermediate regime may fail to admit equilibrium. Importantly, this intermediate region is exactly where deterrence is strategically relevant, because both generals have to take into account both the baseline game and the perceived capture probabilities. However, the defender cannot robustly enforce deterrence while the attacker cannot safely ignore it, which is what causes equilibria to fail to exist in some parameterizations.

\begin{proof}
\noindent\emph{Proof of 1).}
By the definition of $\Gamma$, there exists a defender strategy $Y^*$ satisfying
\[
\min_X X^\top A Y^*=\Gamma.
\]
If $\theta<\Gamma$, then $X^\top A Y^*>\theta$ for every attacker strategy $X$. Hence all attacking troops defect, so the defender's utility is $U_{\mcd,\ea}(X,Y^*)=1$ for all $X$. Since $1$ is the maximum possible defender payoff, $Y^*$ is a best response to every attacker strategy, and therefore $(X,Y^*)$ is a Nash equilibrium for any $X$.

\medskip
\noindent\emph{Proof of 2).}
By the definition of $\Lambda$, there exists $X_B\in\mathcal{E}_B$ such that
\[
\max_Y X_B^\top A Y=\Lambda\le \theta.
\]
Therefore $X_B^\top A Y\le \theta$ for every defender strategy $Y$, so deterrence never occurs when the attacker plays $X_B$. Along this strategy, the ex-ante payoff reduces to
\[
U_{\mcd,\ea}(X_B,Y)=X_B^\top B Y.
\]
Thus the defender may as well play a best response in the baseline game, and the attacker can safely use the baseline equilibrium strategy $X_B$ without worrying about deterrence. Hence $(X_B,Y_B)$ remains a Nash equilibrium of the ex-ante game.

\medskip
\noindent\emph{Proof of 3).}
Consider the symmetric case $N_a=N_d=2$, for which
\begin{equation}
A=
\begin{bmatrix}
1 & 1/2 & 0 \\
1/2 & 1 & 1/2 \\
0 & 1/2 & 1
\end{bmatrix},
\qquad
B=
\begin{bmatrix}
1 & 0 & 0 \\
0 & 1 & 0 \\
0 & 0 & 1
\end{bmatrix}.
\label{eq:example_matrices}
\end{equation}
Then $\Gamma=1/2$ and $\Lambda=2/3$, so the intermediate region is
\[
\frac{1}{2}\le \theta < \frac{2}{3}.
\]
We show that for every $\theta\in[1/2,2/3)$, the ex-ante game has no mixed strategy Nash equilibrium.
Because this $N_a=N_d=2$ instance is invariant under exchanging the two battlefields, any nonsymmetric mixed strategy from either general would favor one battlefield, and the opposing general could exploit this by shifting probability toward or away from the favored side. Thus any equilibrium candidate must assign equal mass to the two endpoint allocations, so we consider strategies of the form
\begin{equation}
X=
\begin{bmatrix}
(1-x)/2\\
x\\
(1-x)/2
\end{bmatrix},
\qquad
Y=
\begin{bmatrix}
(1-y)/2\\
y\\
(1-y)/2
\end{bmatrix},
\label{eq:symmetric_parameterization}
\end{equation}
with $x,y\in[0,1]$.

Under this parameterization, the expected capture fraction is
\[
X^\top A Y=\frac{1+xy}{2},
\]
and the baseline utility is
\begin{equation}
\begin{aligned}
X^\top B Y
&=
y\left(\frac{3x-1}{2}\right)+\frac{1-x}{2} \\
&=
x\left(\frac{3y-1}{2}\right)+\frac{1-y}{2}.
\end{aligned}
\label{eq:example_baseline}
\end{equation}

At the endpoint $\theta=1/2$, any equilibrium candidate must have $x=0$; otherwise the defender can choose any $y>0$ and trigger deterrence. Given $x=0$, the defender uniquely maximizes the baseline payoff by choosing $y=0$. But against $y=0$, deterrence is inactive for every $x$, and the attacker uniquely minimizes the defender payoff by choosing $x=1$, a contradiction. Hence no equilibrium exists at $\theta=1/2$.

It remains to consider $\theta\in(1/2,2/3)$, where we compare the attacker and defender security values. Since every equilibrium payoff in a zero-sum game must equal both security values, strict inequality between them will rule out the existence of equilibrium.

\paragraph*{Attacker security value.}
Define
\[
V_A(\theta):=\min_X \max_Y U_{\mcd,\ea}(X,Y).
\]
Fix $x$. Since the expected capture fraction is increasing in $y$, the defender maximizes capture by choosing $y=1$. To avoid deterrence, the attacker security policy must therefore satisfy
\[
\frac{1+x}{2}\le \theta,
\qquad\text{equivalently}\qquad
x\le 2\theta-1.
\]
Because $\theta<2/3$, any such $x$ satisfies $x<1/3$.

For $x<1/3$, the coefficient of $y$ in \eqref{eq:example_baseline} is negative, so conditional on avoiding deterrence the defender's best response is $y=0$. The resulting payoff is
\[
U_{\mcd,\ea}(X,Y)=\frac{1-x}{2},
\]
which is decreasing in $x$. Hence the attacker minimizes the defender's payoff by taking the largest feasible value,
\[
x^*=2\theta-1.
\]
Substituting yields
\begin{equation}
V_A(\theta)=1-\theta.
\label{eq:attacker_security_value}
\end{equation}

\paragraph*{Defender security value.}
Define
\[
V_D(\theta):=\max_Y \min_X U_{\mcd,\ea}(X,Y).
\]
Fix $y$. To avoid deterrence, the attacker must satisfy
\[
\frac{1+xy}{2}\le \theta,
\qquad\text{equivalently}\qquad
xy\le 2\theta-1.
\]
We consider three cases.

\emph{Case 1: $y\ge 1/3$.}
Then $x=0$ is feasible, and the coefficient of $x$ in \eqref{eq:example_baseline} is nonnegative. Thus $x=0$ minimizes the baseline payoff, giving
\[
U_{\mcd,\ea}(X,Y)=\frac{1-y}{2}.
\]
This is decreasing in $y$, so the best choice within this case is $y=1/3$, which yields $V_D(\theta)=1/3$.

\emph{Case 2: $y\le 2\theta-1$.}
Then the deterrence constraint is satisfied for every $x\in[0,1]$. Since $y<1/3$ throughout this case, the coefficient of $x$ in \eqref{eq:example_baseline} is negative, so the attacker minimizes by choosing $x=1$. The resulting payoff is
\[
U_{\mcd,\ea}(X,Y)=y.
\]
The defender therefore chooses the largest feasible value, $y=2\theta-1$, giving $V_D(\theta)=2\theta-1$. Because $\theta<2/3$, this value is strictly less than $1/3$, so Case 2 is dominated by Case 1.

\emph{Case 3: $2\theta-1<y<1/3$.}
Again the coefficient of $x$ in \eqref{eq:example_baseline} is negative, so the attacker minimizes by choosing the largest feasible $x$, namely
\[
x=\frac{2\theta-1}{y}.
\]
Substituting into \eqref{eq:example_baseline} gives
\[
U_{\mcd,\ea}(X,Y)
=
\frac{6\theta-2-y}{2}
+
\frac{1-2\theta}{2y}.
\]
Differentiating with respect to $y$ shows that this expression is maximized at
\[
y=\sqrt{2\theta-1},
\]
with corresponding payoff
\[
V_D(\theta)=3\theta-1-\sqrt{2\theta-1}.
\]
Setting this equal to $1/3$ shows that Case 3 is preferable to Case 1 when $\theta<5/9$.

Combining the three cases,
\begin{equation}
V_D(\theta)=
\begin{cases}
3\theta-1-\sqrt{2\theta-1}, & \text{if } \dfrac{1}{2} < \theta < \dfrac{5}{9}, \\[8pt]
\dfrac{1}{3}, & \text{if } \dfrac{5}{9}\le \theta < \dfrac{2}{3}.
\end{cases}
\label{eq:defender_security_piecewise}
\end{equation}

Finally, compare the two security values. For $1/2<\theta<5/9$,
\[
V_A(\theta)-V_D(\theta)=2-4\theta+\sqrt{2\theta-1}>0,
\]
and for $5/9\le \theta<2/3$,
\[
V_A(\theta)-V_D(\theta)=\frac{2}{3}-\theta>0.
\]
Hence
\[
V_D(\theta)<V_A(\theta)
\]
for every $\theta\in(1/2,2/3)$. Therefore no mixed strategy Nash equilibrium exists in this region.
\end{proof}

As shown by the proof, equilibria can fail to exist when capture incentives and baseline utility become decoupled. In the intermediate region, $\theta$ is too high for the defender to guarantee deterrence but too low for the attacker to safely minimize defender baseline utility. Both generals must therefore balance two objectives: optimizing the baseline payoff and controlling the expected capture fraction. Because the deterrence decision depends on the mixed strategies themselves, these objectives need not admit equilibrium mixtures that balance both incentives. In a repeated or adaptive setting, play may settle on security policies in which at least one general is not best responding, or it may fail to settle as the generals continually adapt to one another. The dynamics would depend heavily on how attacking troops form or update beliefs about the generals' mixtures.

\section{Ex-Post Decision Point}

We now study the ex-post model already defined in Section~II. In particular, the payoff matrix $M_\ep$ is given by \eqref{eq:ex_post_matrix} and the defender's utility by \eqref{eq:ex_post_utility}. The goal of this section is to analyze the structure of $M_\ep$ and the resulting equilibrium properties.

\subsection{Matrix Structure}

To make the structure of $M_\ep$ explicit, define
\[
\begin{aligned}
J_1(i)&=
\begin{cases}
-1, & i=0,\\
\lfloor \theta i \rfloor, & 1\le i\le N_a,
\end{cases}
\\[4pt]
J_2(i)&=
\begin{cases}
\lceil N_d-\theta(N_a-i)\rceil, & 0\le i\le N_a-1,\\
N_d+1, & i=N_a.
\end{cases}
\end{aligned}
\]
Then the ex-post decision rule from \eqref{eq:ex_post_matrix} implies that, in row $i$, the attacker wins the left battlefield for all columns $j\le J_1(i)$ and the right battlefield for all columns $j\ge J_2(i)$.

Accordingly, row $i$ of $M_\ep$ has entry $-1$ exactly on columns satisfying $J_2(i)\le j\le J_1(i)$, entry $1$ exactly on columns satisfying $J_1(i)<j<J_2(i)$, and entry $0$ elsewhere.

A $1$-entry requires an integer $j$ satisfying
\[
\theta i<j<N_d-\theta(N_a-i),
\]
which is possible only if $\theta<N_d/N_a$. Similarly, a $-1$-entry requires an integer $j$ satisfying
\[
N_d-\theta(N_a-i)\le j\le \theta i,
\]
which is possible only if $\theta\ge N_d/N_a$.

This is the same threshold identified in the ex-interim model in the previous paper \cite{original_blotto_deterrence}, which shows it is never possible for the defender to win both battlefields and for the attacker to win both battlefields in the same game parameterization.

However, the matrix geometry is different from that of the ex-interim model, whose payoff matrix is Toeplitz with a diagonal band of constant width. In the ex-post model, the cutoffs $J_1(i)$ and $J_2(i)$ vary with the row, so the location and width of the nonzero region change from row to row, even though the sign type is globally determined by whether $\theta$ lies below, above, or exactly at $N_d/N_a$.

For example, when $N_a=N_d=3$ and $\theta=0.5<N_d/N_a$, the ex-interim and ex-post models yield
\[
\begin{bmatrix}
1 & 1 & 0 & 0 \\
1 & 1 & 1 & 0 \\
0 & 1 & 1 & 1 \\
0 & 0 & 1 & 1
\end{bmatrix}
\qquad\text{and}\qquad
\begin{bmatrix}
1 & 1 & 0 & 0 \\
0 & 1 & 0 & 0 \\
0 & 0 & 1 & 0 \\
0 & 0 & 1 & 1
\end{bmatrix},
\]
respectively.

\subsection{Comparison with Ex-Interim}

The next result shows that this change in matrix geometry has a consistent directional effect relative to the ex-interim model.

\begin{proposition}\label{prop}
The following hold:
\begin{enumerate}
    \item If $\theta < N_d/N_a$, then $M_\ep \le M_\ei$ entrywise.
    \item If $\theta \ge N_d/N_a$, then $M_\ep \ge M_\ei$ entrywise.
\end{enumerate}
\end{proposition}

\begin{proof}
Fix a pure strategy pair $(i,j)$. By \eqref{eq:Aij},
\[
A_{i,j}
=
\frac{\min\{i,j\}}{N_a}
+
\frac{\min\{N_a-i,N_d-j\}}{N_a},
\]
which is the capture fraction, equivalently the troop-weighted average of the ex-post capture probabilities on the two battlefields.

Suppose first that $\theta < N_d/N_a$. If $M_{\ep,i,j}=1$, then the attacking troops defect on both battlefields in the ex-post model, so both local capture probabilities are strictly greater than $\theta$. Their weighted average is therefore also strictly greater than $\theta$, which implies $A_{i,j}>\theta$ and hence $M_{\ei,i,j}=1$. Since in this regime both matrices have only $0$ and $1$ entries, it follows that $M_\ep \le M_\ei$ entrywise.

Now suppose that $\theta \ge N_d/N_a$. If $M_{\ep,i,j}=-1$, then the attacking troops attack on both battlefields in the ex-post model, so both local capture probabilities are at most $\theta$. Their weighted average is therefore also at most $\theta$, so $A_{i,j}\le\theta$. Because $\theta<1$, attacking on both battlefields also implies $i>j$ and $N_a-i>N_d-j$, and therefore $M_{\ei,i,j}=-1$. Since in this regime both matrices have only $0$ and $-1$ entries, it follows that $M_\ep \ge M_\ei$ entrywise.
\end{proof}

Thus, below the threshold $\theta=N_d/N_a$, ex-post information can only remove defender-favorable $1$-entries, while above the threshold it can only remove attacker-favorable $-1$-entries. Equivalently, the sign of $\theta-N_d/N_a$ determines whether revealing more local information helps or hurts the defender, so this threshold can be interpreted as a simple rule that could be used by either general for information design. For example, the attacker may be able to control the level of information by exposing or withholding their plans, while the defender may be able to control the level of information by making it easier or harder for the attacking troops to observe local conditions and defect during later stages of an attack.

\subsection{Equilibrium Existence and Implications}

The ex-post model is represented by the finite payoff matrix $M_\ep$, so it is a standard zero-sum game. Hence a mixed strategy Nash equilibrium always exists, and the value of the game satisfies
\begin{equation}
v=\min_X \max_Y X^\top M_\ep Y
=\max_Y \min_X X^\top M_\ep Y.
\label{eq:later_minimax}
\end{equation}

Equilibrium strategies can be computed by linear programming. The attacker's problem is
\begin{align*}
\min_{X,v}\quad & v \\
\text{s.t.}\quad & X^\top M_\ep e_j \le v, \qquad j=0,1,\dots,N_d, \\
& \sum_{i=0}^{N_a} X_i = 1, \\
& X_i \ge 0, \qquad i=0,1,\dots,N_a,
\end{align*}
where $e_j$ denotes the $j$th pure strategy of the defender.

The defender's dual problem is
\begin{align*}
\max_{Y,w}\quad & w \\
\text{s.t.}\quad & e_i^\top M_\ep Y \ge w, \qquad i=0,1,\dots,N_a, \\
& \sum_{j=0}^{N_d} Y_j = 1, \\
& Y_j \ge 0, \qquad j=0,1,\dots,N_d,
\end{align*}
where $e_i$ denotes the $i$th pure strategy of the attacker.

By strong duality, the optimal objective values coincide, so $v=w$ equals the value of the game. Thus, once the matrix $M_\ep$ from \eqref{eq:ex_post_matrix} has been specified, equilibrium strategies for both generals can be obtained directly by solving this primal-dual pair.

Let $v_\ei$ and $v_\ep$ denote the defender equilibrium values of $\mathrm{CB}_\ei(N_a,N_d,\theta)$ and $\mathrm{CB}_\ep(N_a,N_d,\theta)$, respectively. Combining the comparison in the previous subsection with the minimax formulation gives
\[
v_\ep \le v_\ei \quad \text{if } \theta < N_d/N_a,
\]
and
\[
v_\ep \ge v_\ei \quad \text{if } \theta \ge N_d/N_a.
\]
The threshold $N_d/N_a$ has a simple interpretation: up to the natural cap at one, it is the largest aggregate fraction of attacking troops that can be captured. When $\theta<N_d/N_a$, the attacking troops are unwilling to tolerate this aggregate capture fraction, so ex-post information can only hurt the defender and help the attacker: some troops who otherwise would have defected may attack upon discovering that their individual assignment is favorable. When $\theta\ge N_d/N_a$, even the worst-case aggregate capture fraction is tolerable, so ex-post information can only help the defender and hurt the attacker: revealing local conditions may cause troops on unusually risky battlefields to defect when they otherwise would have attacked.

\section{Numeric Results}

This section explores examples where we solve for equilibrium and security utilities as a function of $\theta$ across different information structures. For the ex-ante decision point, we use the security values defined in the proof from section III for the two troop symmetric case, as well as the values for the three troop symmetric case which were computed in the same way. For the ex-interim and ex-post decision points, we constructed the payoff matrices $M_\ei$ and $M_\ep$ according to \eqref{eq:ex_interim_matrix} and \eqref{eq:ex_post_matrix}, respectively, and solved for equilibrium utilities using the linear programming formulation described in the previous section.

\subsection{Ex-Ante Decision Point}

The examples below compare the ex-ante model to the ex-interim model when both sides have two or three troops. We can visualize the theorem given in section III: for $\theta<\Gamma$, deterrence can be forced; for $\theta\ge\Lambda$, a baseline equilibrium survives; and in the intermediate region the security values are not guaranteed to coincide.

\noindent\textbf{Example 1: $N_a=2$, $N_d=2$.} Fig.~\ref{fig:early_2v2}.

\begin{figure}[H]
    \centering
    \includegraphics[width=\columnwidth]{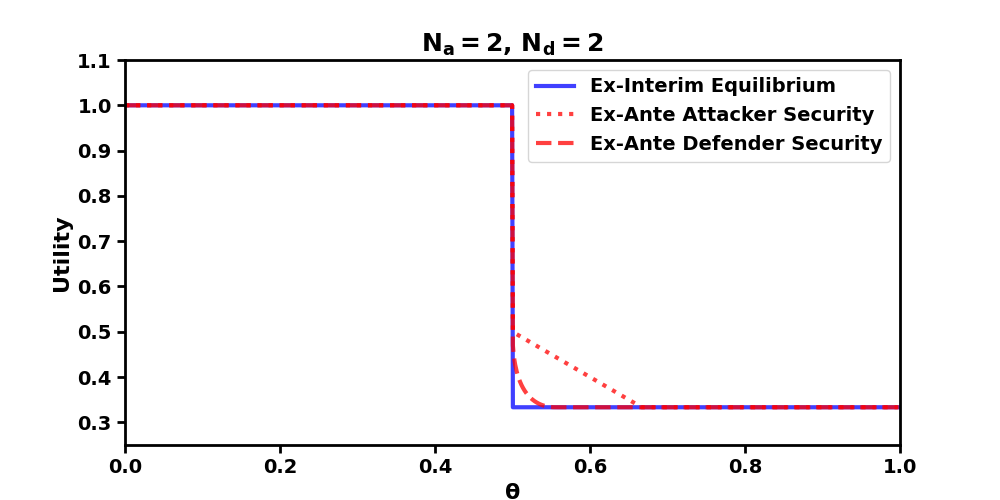}
    \caption{Ex-interim equilibrium utility and ex-ante security values for $N_a=2$, $N_d=2$.}
    \label{fig:early_2v2}
\end{figure}

The ex-ante and ex-interim models have the same value for $\theta<1/2$ (deterrence is guaranteed) and for $\theta\ge 2/3$ (baseline equilibrium is maintained). The difference appears on $1/2\le \theta < 2/3$. For the ex-interim model, this region has the same value as the baseline game. For the ex-ante model, the security values diverge. Thus, in this example, the ex-ante information structure favors the defender throughout the intermediate region.

\noindent\textbf{Example 2: $N_a=3$, $N_d=3$.} Fig.~\ref{fig:early_3v3}.

\begin{figure}[H]
    \centering
    \includegraphics[width=\columnwidth]{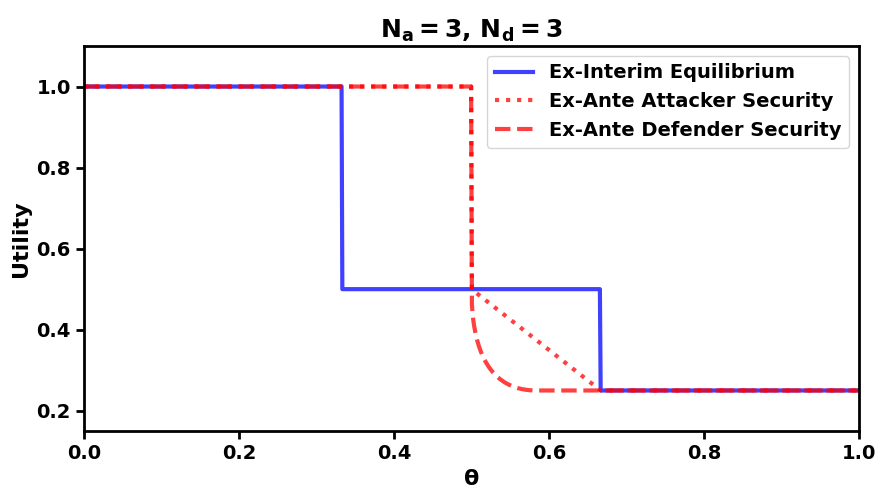}
    \caption{Ex-interim equilibrium utility and ex-ante security values for $N_a=3$, $N_d=3$.}
    \label{fig:early_3v3}
\end{figure}

The comparison is quite different when both sides have three troops instead of two. When $1/3\le \theta < 1/2$, both models have equilibria, but the ex-interim model no longer gets full deterrence, meaning the defender prefers the ex-ante model in this region of $\theta$. However, once $\theta$ crosses into the region from $1/2$ to $2/3$, the ex-ante model no longer has an equilibrium and both security values drop \emph{below} the ex-interim utility, meaning the \emph{attacker} now prefers the ex-ante model. This example shows that neither decision point uniformly dominates the other.

\subsection{Ex-Post Decision Point}

We now compare the ex-post model with the ex-interim model. These examples illustrate the result from Section~IV, showing that the threshold $\theta=N_d/N_a$ determines which information structure is preferable for each player.

\noindent\textbf{Example 3: $N_a=4$, $N_d=2$.} Fig.~\ref{fig:later_4v2}.

\begin{figure}[H]
    \centering
    \includegraphics[width=\columnwidth]{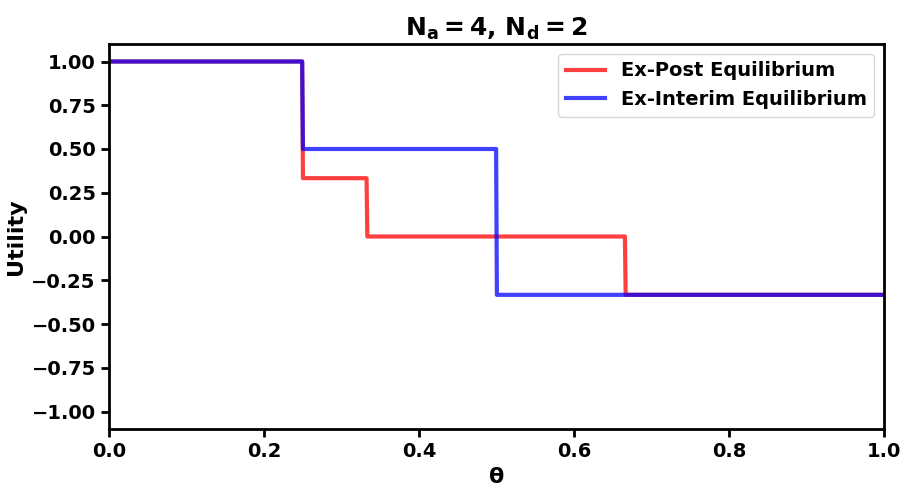}
    \caption{Ex-interim and ex-post equilibrium utilities for $N_a=4$, $N_d=2$.}
    \label{fig:later_4v2}
\end{figure}

This example is interesting because $N_d/N_a=1/2$ lies in the interior of the interval of $\theta$ values, so the preferred information structure flips within the figure. Thus, this example illustrates the preference reversal at the threshold $\theta=N_d/N_a$.

\noindent\textbf{Example 4: $N_a=100$, $N_d=40$.} Fig.~\ref{fig:later_100v40}.

\begin{figure}[H]
    \centering
    \includegraphics[width=\columnwidth]{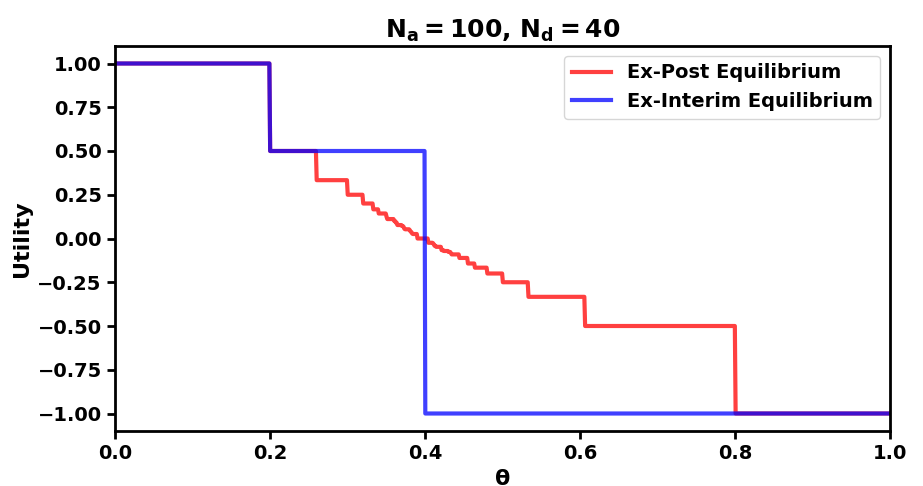}
    \caption{Ex-interim and ex-post equilibrium utilities for $N_a=100$, $N_d=40$.}
    \label{fig:later_100v40}
\end{figure}

Even when the number of troops is higher, the threshold $N_d/N_a=2/5$ still determines which player prefers each information structure. The equilibrium utilities take on a larger set of values, but the same threshold rule still determines the preference ordering.

\noindent\textbf{Example 5: $N_a=100$, $N_d=100$.} Fig.~\ref{fig:later_100v100}.

\begin{figure}[H]
    \centering
    \includegraphics[width=\columnwidth]{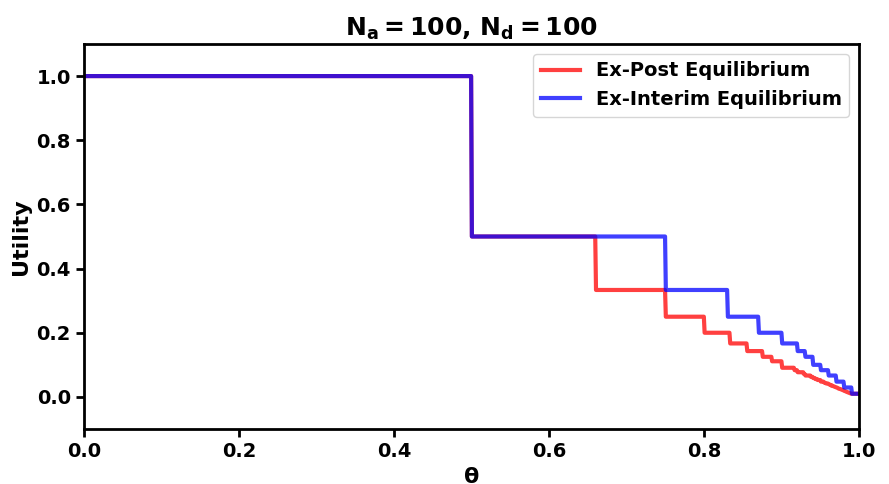}
    \caption{Ex-interim and ex-post equilibrium utilities for $N_a=100$, $N_d=100$.}
    \label{fig:later_100v100}
\end{figure}

Here $N_d/N_a=1$, so the threshold lies at the boundary of the admissible interval $\theta\in[0,1)$. Consequently, the player preferences between the ex-interim and ex-post models do not reverse anywhere in the figure. This example highlights the boundary case where ex-interim information weakly favors the defender throughout $\theta\in[0,1)$.

\section{Conclusion}

In this paper, we studied how information structure affects deterrence in a two-battlefield Colonel Blotto game. We showed that these different structures lead to fundamentally different strategic interactions at the level of the generals: in the ex-ante model, equilibrium may fail to exist, whereas in the ex-post model equilibrium remains well-defined but the induced payoff matrix changes structure. We also showed that none of these information structures uniformly dominates the others, indicating the importance of understanding which structure is being used in real-world applications. In the case of the ex-interim and ex-post models, we showed that the threshold $\theta=N_d/N_a$ determines which information structure is preferable for each player.

The fact that the different models yield such different outcomes, and that information can often be manipulated in real-world settings, suggests that future work should consider how to design information structures in these games in order to achieve desired outcomes. We would also like to explore how these results extend to more than two battlefields, and to other classes of games where information structure may play a critical role in strategic interactions.

\bibliographystyle{IEEEtran}
\bibliography{references}

@article{pita2009using,
  author  = {James Pita and Manish Jain and Fernando Ord{\'o}{\~n}ez and Christopher Portway and Milind Tambe and Craig Western and Praveen Paruchuri and Sarit Kraus},
  title   = {Using Game Theory for {Los Angeles} Airport Security},
  journal = {AI Magazine},
  year    = {2009},
  volume  = {30},
  number  = {1},
  pages   = {43--57},
  doi     = {10.1609/aimag.v30i1.2173}
}

@book{schelling1960strategy,
  author    = {Thomas C. Schelling},
  title     = {The Strategy of Conflict},
  publisher = {Harvard University Press},
  year      = {1960}
}

@book{powell1990nuclear,
  author    = {Robert Powell},
  title     = {Nuclear Deterrence Theory: The Search for Credibility},
  publisher = {Cambridge University Press},
  year      = {1990},
  doi       = {10.1017/CBO9780511551598}
}

@article{borel1921theorie,
  author  = {{\'E}mile Borel},
  title   = {La th{\'e}orie du jeu et les {\'e}quations int{\'e}grales {\`a} noyau sym{\'e}trique},
  journal = {Comptes rendus hebdomadaires des s{\'e}ances de l'Acad{\'e}mie des sciences},
  year    = {1921},
  volume  = {173},
  pages   = {1304--1308}
}

@techreport{gross1950continuous,
  author      = {Oliver Alfred Gross and R. A. Wagner},
  title       = {A Continuous {Colonel Blotto} Game},
  institution = {RAND Corporation},
  year        = {1950},
  number      = {RM-408}
}

@article{roberson2006colonel,
  author  = {Brian Roberson},
  title   = {The {Colonel Blotto} Game},
  journal = {Economic Theory},
  year    = {2006},
  volume  = {29},
  number  = {1},
  pages   = {1--24},
  doi     = {10.1007/s00199-005-0071-5}
}

@article{hart2008discrete,
  author  = {Sergiu Hart},
  title   = {Discrete {Colonel Blotto} and {General Lotto} Games},
  journal = {International Journal of Game Theory},
  year    = {2008},
  volume  = {36},
  number  = {3--4},
  pages   = {441--460},
  doi     = {10.1007/s00182-007-0099-9}
}

@article{kovenock2011multidim,
  author  = {Dan Kovenock and Brian Roberson},
  title   = {A {Blotto} Game with Multi-Dimensional Incomplete Information},
  journal = {Economics Letters},
  year    = {2011},
  volume  = {113},
  number  = {3},
  pages   = {273--275},
  doi     = {10.1016/j.econlet.2011.08.009}
}

@article{ahmadinejad2019duels,
  author  = {AmirMahdi Ahmadinejad and Sina Dehghani and MohammadTaghi Hajiaghayi and Brendan Lucier and Hamid Mahini and Saeed Seddighin},
  title   = {From Duels to Battlefields: Computing Equilibria of {Blotto} and Other Games},
  journal = {Mathematics of Operations Research},
  year    = {2019},
  volume  = {44},
  number  = {4},
  pages   = {1304--1325},
  doi     = {10.1287/moor.2018.0971}
}

@article{shubik1981systems,
  author  = {Martin Shubik and Robert J. Weber},
  title   = {Systems Defense Games: {Colonel Blotto}, Command and Control},
  journal = {Naval Research Logistics Quarterly},
  year    = {1981},
  volume  = {28},
  number  = {2},
  pages   = {281--287},
  doi     = {10.1002/nav.3800280210}
}

@article{kovenock2021generalizations,
  author  = {Dan Kovenock and Brian Roberson},
  title   = {Generalizations of the {General Lotto} and {Colonel Blotto} Games},
  journal = {Economic Theory},
  year    = {2021},
  volume  = {71},
  number  = {3},
  pages   = {997--1032},
  doi     = {10.1007/s00199-020-01272-2}
}

@article{anbarci2023dynamicblotto,
  author  = {Nejat Anbarci and Kutay Cingiz and Mehmet S. Ismail},
  title   = {Proportional Resource Allocation in Dynamic $n$-Player {Blotto} Games},
  journal = {Mathematical Social Sciences},
  year    = {2023},
  volume  = {125},
  pages   = {94--100},
  doi     = {10.1016/j.mathsocsci.2023.07.002}
}

@article{paarporn2025incomplete,
  author  = {Keith Paarporn and Rahul Chandan and Mahnoosh Alizadeh and Jason R. Marden},
  title   = {Incomplete and Asymmetric Information in {General Lotto} Games},
  journal = {IEEE Transactions on Automatic Control},
  year    = {2025},
  volume  = {70},
  number  = {6},
  pages   = {3617--3632},
  month   = jun,
  doi     = {10.1109/TAC.2024.3506786}
}

@article{paarporn2024intentions,
  author  = {Keith Paarporn and Rahul Chandan and Dan Kovenock and Mahnoosh Alizadeh and Jason R. Marden},
  title   = {Strategically Revealing Intentions in {General Lotto} Games},
  journal = {IEEE Transactions on Automatic Control},
  year    = {2024},
  volume  = {69},
  number  = {8},
  pages   = {5396--5407},
  doi     = {10.1109/TAC.2024.3367651}
}

@inproceedings{original_blotto_deterrence,
  author    = {David Grimsman and Keith Paarporn},
  title     = {A {Colonel Blotto} Approach to Deterrence},
  booktitle = {2025 IEEE 64th Conference on Decision and Control (CDC)},
  year      = {2025},
  pages     = {4416--4421},
  address   = {Rio de Janeiro, Brazil},
  month     = dec,
  doi       = {10.1109/CDC57313.2025.11312783}
}

\end{document}